\documentclass[letterpaper, 10 pt, conference]{ieeeconf}
\usepackage{xcolor}
\usepackage{float}
\usepackage{multirow}
\usepackage{amsmath}
\usepackage{amssymb,eqnarray}
\usepackage{mathabx}
\usepackage{graphicx}
\usepackage{stmaryrd}
\usepackage{psfrag}
\usepackage{pstool}
\usepackage{bbm}
\usepackage{mathtools}
\usepackage{graphics} 
\usepackage{epsfig}
\usepackage{tabu,multirow}
\usepackage{epstopdf}

\newtheorem{thm}{Theorem}
\newtheorem{lem}[thm]{Lemma}

\newtheorem{rmk}{Remark}
\newtheorem{assumption}{Assumption}
\newtheorem{defn}{Definition}
\newtheorem{problem}{Problem}

\usepackage{tikz}
\usetikzlibrary{automata, positioning, arrows}
\tikzset{
->, 
>=latex,
node distance=3.5cm, 
every state/.style={thick, fill=blue!15}, 
initial text=$ $, 
}

\makeatletter
\newcommand{\pushright}[1]{\ifmeasuring@#1\else\omit\hfill$\equationstyle#1$\fi\ignorespaces}
\newcommand{\pushleft}[1]{\ifmeasuring@#1\else\omit$\equationstyle#1$\hfill\fi\ignorespaces}
\makeatother

\usepackage{multicol}
\usepackage{booktabs} 
\usepackage[noend]{algorithmic}
\usepackage{algorithm2e}
\usepackage[english]{babel}
\usepackage{subfigure}
\usepackage{mathtools}
\mathtoolsset{showonlyrefs=false}
\makeatletter
\usepackage{multirow}   
\usepackage{array, makecell} %

\usepackage{multicol}
\usepackage[noend]{algorithmic}
\usepackage{algorithm2e}

\providecommand{\remarkname}{\textbf{Remark}}
\providecommand{\theoremname}{\textbf{Theorem}}
\providecommand{\lemmaname}{\textbf{Lemma}}
\providecommand{\assumname}{\textbf{Assumption}}

\IEEEoverridecommandlockouts

\usepackage[margin=.7in]{geometry}
\newcommand{\note}[1]{}

\title{\LARGE \bf{Distributionally Robust Model Predictive Control\\ with Total Variation Distance}}
\author{Anushri Dixit, Mohamadreza Ahmadi, and Joel W. Burdick \thanks{The authors are with the California Institute of Technology, 1200 E. California Blvd., MC 104-44, Pasadena, CA 91125,  e-mail: (\{adixit, mrahmadi\}@caltech.edu, jwb@robotics.caltech.edu.) }}

\begin{document}

\maketitle

\begin{abstract}
This paper studies the problem of distributionally robust model predictive control (MPC) using total variation distance ambiguity sets. For a discrete-time linear system with additive disturbances, we provide a conditional value-at-risk reformulation of the MPC optimization problem that is distributionally robust in the expected cost and chance constraints. The distributionally robust chance constraint is over-approximated as a simpler, tightened chance constraint that reduces the computational burden. Numerical experiments support our results on probabilistic guarantees and computational efficiency.
\end{abstract}

\section{Introduction}
Model {Predictive} Control (MPC) is a widely used method for robot motion planning, because it incorporates state and control constraints in a receding horizon fashion~\cite{borrelli2003constrained,fan2021step}. There are many ways to incorporate uncertainty in MPC. Robust MPC accounts for worst-case disturbances in a set of bounded uncertainties~\cite{bemporad1999robust}. This approach is often too conservative, since it does not account for the distribution of the uncertainties. Stochastic MPC (SMPC)~\cite{mesbah2016SMPC} minimizes the expected value of a cost function, while respecting a bound on the probability of violating state and control constraints (also called chance constraints). Risk-aware MPC methods use coherent risk measures~\cite{dixit2020risksensitive, Sopasakis2019} to account for variations in the underlying distribution of uncertainty.  This is convenient since one often only has an estimate of the true uncertainty distribution. This notion of allowing for variation in the underlying distribution is called distributional robustness. 

Conditional value-at-risk (CVaR) is an important example of a coherent risk measure that has received significant attention in risk-aware MPC. In~\cite{singh2018framework}, the authors proposed a Lyapunov condition for risk-sensitive exponential stability in the presence of discretely quantized process noise for a CVaR objective but did not include risk constraints in their formulation. Measurement noise and moving obstacles were considered in~\cite{hakobyan2019risk}, wherein the authors devised an MPC-based scheme for path planning with CVaR safety constraints when a reference trajectory is generated by RRT$^*$~\cite{karaman2011sampling} and extended to a Wasserstein distributionally robust formulation in~\cite{hakobyan2020wasserstein}. A method based on stochastic reachability analysis was proposed in~\cite{chapman2019risk} to estimate a CVaR-safe set of initial conditions via the solution to a Markov Decision Process. Risk-sensitive obstacle avoidance has been tackled through CVaR control barrier functions in \cite{ahmadi2020cvar} and applied to bipedal robot locomotion. {In this setting, discrete, additive process noise was considered because it could be numerically calculated in simulation for a series of walking behaviors.}

It is important to note that coherent risk measures have heretofore provided distributional robustness in the cost, but not in the chance constraints that one may need to satisfy in SMPC. This work extends the use of coherent risk measures to provide distributional robustness in the chance constraints.

Distributionally robust chance constraints (DRCCs) have been well studied in stochastic optimization. A popular metric for enforcing distributional robustness is the Wasserstein distance. In~\cite{xie2021distributionally} the author proposed a tight inner and outer approximation of the DRCC with a Wasserstein ambiguity set using a CVaR reformulation. In~\cite{mark2020DRCC}, the authors enforced DRCCs with Wasserstein distance ambiguity set in an MPC setting. Optimal control using distributionally robust CVaR constraints with second-order moment ambiguity sets can be posed as a semidefinite program in~\cite{parys2016DRCC}. 

The total variation distance (TVD) is another commonly used bounding metric on probability spaces. Intuitively, it provides an upper bound on the difference of probabilities that an event occurs under two measures~\cite{gibbs02onchoosing}, see Fig.~\ref{fig:illustration_dr} for an illustration of TVD. {In~\cite{ioannis2021tvdlqr}, the authors provide a TVD-based distributionally robust solution of the linear, quadratic regulator and use this formulation for a drop–shipping retail fulfillment application. In~\cite{zolanvari2021data}, the authors consider a data-driven strategy to solve iterative tasks using a MPC scheme. This framework is amenable to general ambiguity sets, including TVD.}

This paper extends the literature on distributionally robust MPC using TVD ambiguity sets.  We provide a deterministic approximation of a stochastic MPC optimization problem with a distributionally robust objective and DRCC with TVD ambiguity sets for a discrete distribution. This is achieved by over-approximating the DRCC in the form of a simple, but more conservative, chance constraint that is further simplified using CVaR. The objective is also reformulated as a CVaR objective. The resulting MPC optimization is an efficient quadratic program.
\begin{figure}[t]
\centerline{\includegraphics[width=80mm,scale=0.5]{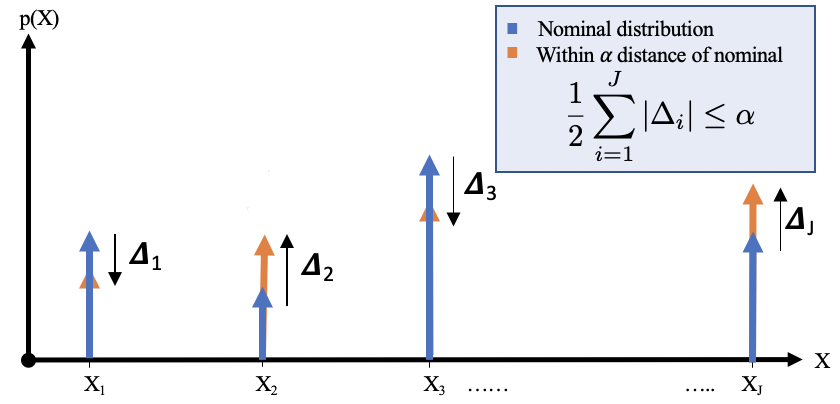}}
\vskip -0.1true in
\caption{Depiction of the Total Variation Distance (TVD) between discrete probability distributions. Our method is robust to any distribution within TVD $ \leq \alpha$ of a nominal distribution, where $\alpha\in(0,1)$.}  \label{fig:illustration_dr}
\vskip -0.36 true in
\end{figure}
\section{Preliminaries} \label{sec:preliminaries}
Consider a probability space $(\Omega, \mathcal{F}, {P})$, where $\Omega$, $\mathcal{F}$, and $P$ are the sample space, $\sigma$-algebra over $\Omega$, and probability measure over $\mathcal{F}$, respectively. A random variable $C: \Omega\xrightarrow[]{}\mathbb{R}$ denotes the cost of each outcome. For a constant $p \in [1,\infty)$, $\mathcal{L}_p(\Omega, \mathcal{F},P)$ denotes the vector space of real valued random variables $C$ for which $\mathbb{E}[|C|^p] < \infty$. The set of all random variables defined on $\Omega$ is  $\mathcal{C}$. A risk measure is a function that maps a cost random variable to a real number, $\rho:\mathcal{C}\xrightarrow[]{}\mathbb{R}$. 

For constrained stochastic optimization programs, chance constraints can be reformulated by a commonly used risk measure called the \textit{Value-at-Risk} (VaR). For a given confidence level $\alpha \in (0,1)$, $\mathrm{VaR}_{1-\alpha}$ denotes the $({1-\alpha})$-quantile value of the cost variable $C$ and is defined as, 
\begin{align*}
    \text{VaR}_{1-\alpha}(C) := \inf \{ z \,|\, \mathbb{P}(C \leq z) \geq \alpha\}.
\end{align*}
It follows that $
    \text{VaR}_{1-\alpha}(C)\leq 0 \implies \mathbb{P}(C \leq 0)\geq\alpha.$
However, VaR is generally nonconvex and hard to compute. We now introduce convex and monotonic risk measures.
\begin{defn}[Coherent Risk Measures{~\cite{artzner1999coherent, tamar2015policy}}]{Consider two random variables, {$C$, $C'\in\mathcal{C}$}. A coherent risk measure, $\rho:\mathcal{C}\xrightarrow[]{}\mathbb{R}$, is a risk measure that satisfies the following properties:
\begin{enumerate}
    \item \textbf{Monotonicity} $C \leq C' \implies \rho(C) \leq \rho(C')$,
    \item \textbf{Translational invariance} $\rho(C + a) = \rho(C) + a, \, \forall a \in \mathbb{R}$,
    \item \textbf{Positive homogeneity} $\rho(a C) = a\rho(C), \, \forall a \geq 0$, 
    \item \textbf{Subadditivity} $\rho (C + C') \leq \rho(C) + \rho(C')$.
\end{enumerate}}
\end{defn}

Coherent risk measures can be written as the worst-case expectation over a convex and closed set of probability mass (or density) functions (pdfs/pmfs). This convex set in the dual representation of a risk measure is an \textit{ambiguity set}.

\begin{defn}[Representation Theorem]{Every coherent risk measure can be represented in its dual form as, 
\begin{align*}
    \rho(C) := \sup_{Q \in \mathcal{Q}}E_Q(C),
\end{align*}

where the ambiguity set (or risk envelope) is convex and closed,  {and probability density $Q(X)$ is absolutely continuous with respect to density $P(X)$; \textit{i.e.,} $P(X)=0\ \rightarrow \ Q(X)=0$.}
}
\end{defn}

While coherent risk measures act on scalar cost random variables, in this paper, we write $\rho(\boldsymbol{C})$, where $\boldsymbol{C}$ is an $n-$vector of cost random variables, to mean $\rho(\boldsymbol{C}) = \begin{bmatrix}\rho(C_1), \dotsc, \rho(C_n) \end{bmatrix}^T$. We present two examples of coherent risk measures pertinent to the problem presented in this work.

\subsection{Conditional Value-at-Risk}

The {\em conditional value-at-risk}, $\mathrm{CVaR}_{1-\alpha}$, measures the expected loss in the $({1-\alpha})$-tail given that the threshold $\mathrm{VaR}_{1-\alpha}$ has been crossed. $\mathrm{CVaR}_{1-\alpha}$ is found as{~\cite{rockafellar2000optimization}
\begin{equation}
\begin{aligned}
    \mathrm{CVaR}_{1-\alpha}(C):=&\inf_{z \in \mathbb{R}}\mathbb{E}\Bigg[z + \frac{(C-z)^{+}}{1-\alpha}\Bigg], \\ =& \inf_{z \in \mathbb{R}, s \in \mathbb{R}^J} \sum_{j=1}^{J}p(j) \bigg(z +\frac{s(j)}{1-\alpha} \bigg) \\
    & \quad\text{s.t.} \quad s(j) \geq 0, \, s(j) + z \geq C^j,
\end{aligned}
\end{equation} 
}
where $(\cdot)^{+}:=\max\{\cdot, 0\}$ and $C^j$ denotes the $j^{\text{th}}$ realization of the random variable $C$ whose probability of occurrence is given by $p(j)$. A value of $\alpha \simeq 0$ corresponds to a risk-neutral case. A value of $\alpha \to 1$ is rather a risk-averse case.  
CVaR provides a convex upper bound of VaR, i.e.,
\begin{align}
    \small{\text{VaR}_{1-\alpha}(C) \leq \text{CVaR}_{1-\alpha}(C) \leq 0 \implies \mathbb{P}(C\leq 0)\geq \alpha.}
\end{align}
\vskip -0.1 true in
\subsection{Total Variation Distance}
This work aims to achieve distributional robustness for ambiguity sets defined by the total variation distance~\cite{shapiro2017distributionally},
\begin{equation*}
  \displaystyle \text{TVD}_{1-\alpha}(C)\! := \! \sup_{Q \in \mathcal{Q}_\text{TVD}}\!\!\!E_Q(C) = \alpha\sup_{c\in\Omega}c + (1-\alpha) \displaystyle \text{CVaR}_{1-\alpha}(C),
\end{equation*}
where the ambiguity set of $\text{TVD}_{1-\alpha}(C)$ is given by,
\begin{equation*}
    \mathcal{Q}_{\text{TVD}} := \Big\{q \in \Delta_J: d_{TV}(p,q) = \frac{1}{2}\sum_{j=1}^{J}|q(j)-p(j)|\leq \alpha\Big\},
\end{equation*}
where $\Delta_J$ is the probability simplex, $ \Delta_J:= \{q \in \mathbb{R}^J\,|\, q \geq 0, \,\sum_{j=1}^{J}q(j) = 1\}$ and {$p$ refers to the probability mass function associated with the random variable $C$}. Gibbs and Su~\cite{gibbs02onchoosing} survey the relationships between total variation and various other probability metrics/distances. 

\section{Problem formulation}

We consider a class of discrete-time systems given by
\begin{equation}\label{eq:sys}
\begin{aligned}
    \boldsymbol{x}(t+1) &= A\boldsymbol{x}(t) + B\boldsymbol{u}(t) + D\boldsymbol{\delta}(t)
\end{aligned}
\end{equation}
where $\boldsymbol{x}(t) \in \mathbb{R}^{n_x}$ and $\boldsymbol{u}(t) \in \mathbb{R}^{n_u}$ are the system state and controls at time $t$, respectively. The system is affected by a stochastic, additive, process noise $\boldsymbol{{\delta}}_t \in \mathbb{R}^{n_d}$.
\begin{assumption}[Discrete process noise]\label{assumption: process_noise}
\textit{The process noise $\boldsymbol{\delta}$ consists of i.i.d. samples of a discrete distribution given by the probability mass function (pmf), $\boldsymbol{p} = [p(1), p(2), \dotsc, p(J)]^T$}. For this distribution, we also define the index set $\mathcal{D} = \{1, \dotsc, J\}$.\footnote{The results in this paper can be extended to continuous distributions using sample average approximation~\cite{kleywegt2002sample} or other sample-based techniques~\cite{blackmore2010particleSMPC}.}
\end{assumption}
Consider there are {$r$ state constraints} that take the form $$\mathcal{X} := \{\boldsymbol{x} \in \mathbb{R}^{n_x} | F_x\boldsymbol{x} \leq g_x\}, F_x\in \mathbb{R}^{r \times n_x}, g_x \in \mathbb{R}^r.$$ In this work, we are interested in satisfying these state constraints in a distributionally robust manner,
\begin{align*}
   1-\epsilon &\leq \min_{Q \in \mathcal{Q}}\mathbb{P}_Q(\boldsymbol{x} \in \mathcal{X})  \\
     &=1 - \max_{Q \in \mathcal{Q}}\mathbb{P}_Q(\boldsymbol{x} \not\in \mathcal{X}) \\
     &= 1 - \max_{Q \in \mathcal{Q}}\mathbb{E}_Q(\mathbbm{1}_{\boldsymbol{x} \not\in \mathcal{X}}) \\
     &= 1 - \rho(\mathbbm{1}_{\boldsymbol{x} \not\in \mathcal{X}}),
\end{align*}
where we used the property that the probability of an event can be expressed as the expected value of its indicator to express the DRCC as the risk of an indicator. Hence, 
\begin{align}\label{eq:state_const}
    \min_{Q\in\mathcal{Q}_{\text{TVD}}} \mathbb{P}_Q(\boldsymbol{x} \in \mathcal{X})\geq 1-\epsilon \iff \text{TVD}_{1-\alpha}(\mathbbm{1}_{\boldsymbol{x} \not\in \mathcal{X}}) &\leq \epsilon.
\end{align}

 \begin{problem}
\textit{
Consider the discrete-time system~\eqref{eq:sys}. Given a deterministic initial condition $\boldsymbol{x}_0 \in \mathbb{R}^{n_x}$, state constraints $\mathcal{X}\subset \mathbb{R}^{n_x}$, convex polytopic control constraints $\mathcal{U}\subset \mathbb{R}^{n_u}$, a convex stage cost $c:\mathbb{R}^{n_x} \times \mathbb{R}^{n_u}\to \mathbb{R}_{\ge0}$, a horizon $N \in \mathbb{N}$, and risk tolerance $\epsilon\in(0,1)$\ for state constraints, compute the receding horizon controller $\boldsymbol{u} = \{\boldsymbol{u}_k \}_{k=0}^{N-1}$ such that the total cost $\mathcal{J}({x}(t), \boldsymbol{u})$ is minimized while satisfying the distributionally robust constraints~\eqref{eq:state_const}, i.e., the solution to the following optimization problem,}
\begin{subequations}\label{eq:mpc1}
\begin{align}
\begin{split}\label{eq:cost}
\min_{\boldsymbol{u}} \quad &\mathcal{J}(x(t), \boldsymbol{u}) := \text{TVD}_{1-\alpha}\bigg(\sum_{k=0}^{N-1}c(\boldsymbol{x}_k, \boldsymbol{u}_k)\bigg) \quad
\end{split}\\
\begin{split}\label{eq:dyn1}
 \textrm{s.t.} \quad &\boldsymbol{x}_{k+1} = A\boldsymbol{x}_k + B\boldsymbol{u}_k + D{{\boldsymbol{\delta}}_k},
\end{split}\\
\begin{split}\label{eq:stcon}
&\text{TVD}_{1-\alpha}(\mathbbm{1}_{\boldsymbol{x}_k \not\in \mathcal{X}}) \leq \epsilon, 
\end{split}\\
\begin{split}\label{eq:ic}
&\boldsymbol{u}_k \in \mathcal{U}, \,\quad\boldsymbol{x}_0 = \boldsymbol{x}(t),\, \forall k \in \{0, \dotsc N-1 \},
\end{split}
\end{align}
\end{subequations}
\textit{where, $\boldsymbol{x}_k = \boldsymbol{x}(t+k|t)$ and $\boldsymbol{u}_k = \boldsymbol{u}(t+k|t)$.}
\end{problem}

\section{MPC reformulation}
The batch form of Equation~\eqref{eq:dyn1} can be re-written as
\begin{align*}
    \boldsymbol{x}_{k+1} &= A^{k+1}\boldsymbol{x}_0 + \sum_{i = 0}^{k}\Big( A^iB\boldsymbol{u}_{k+1-i}+ A^iD\boldsymbol{\delta}_{k+1-i}\Big),\\
    &=A^{k+1}\boldsymbol{x}_0 + \boldsymbol{B}_{k+1}\boldsymbol{\bar{u}}_{k+1} + \boldsymbol{D}_{k+1}\boldsymbol{\bar{\delta}}_{k+1},
\end{align*}
where, $\boldsymbol{B}_k, \,\boldsymbol{\bar{u}}_k,\, \boldsymbol{D}_k, \text{ and }\boldsymbol{\bar{\delta}}_k$ are described in Appendix~\ref{appendix: batch}.
Clearly, as $k$ increases, the disturbance effects compound. At each step $k$, the distribution of $\boldsymbol{x}_k$ is given by the joint distribution of $(\boldsymbol{\delta_1},\, \dotsc, \, \boldsymbol{\delta_k})$. Let $\boldsymbol{p}_k$ denote the probability mass function of this joint distribution. The size of this joint pmf is $J^k$ (see Assumption 1). 

The following key result provides an over approximation of the DRCC for TVD ambiguity sets.
\begin{lem}[Risk Reformulation]\label{lem:risk_constraint}
If Assumption~\ref{assumption: process_noise} holds, then~\eqref{eq:stcon} is satisfied if the following constraint is satisfied,
\begin{equation}\label{eq:strict_stcon}
\begin{aligned}
    {\mathbb{P}(\boldsymbol{x}_k \not\in \mathcal{X})+ \alpha\leq \epsilon.}
\end{aligned}
\end{equation}
\end{lem}
\begin{proof}
We express the risk constraint through its dual representation. Note that we are finding the worst-case expectation within the TVD-based risk envelope. The distribution that gives us this worst-case expectation has the pmf $\boldsymbol{q}\circ\boldsymbol{p}_k$ (where $\circ$ denotes element-wise multiplication) in the following optimization,
\begin{subequations}
\begin{align}
\begin{split}\label{eq:opt1}
    \text{TVD}_{1-\alpha}(\mathbbm{1}_{\boldsymbol{x}_k \not\in \mathcal{X}}) &= \max_{\boldsymbol{q}}  \sum_{j = 1}^{J^k} q(j)p_k(j) \mathbbm{1}_{\boldsymbol{x}^j_k \not\in \mathcal{X}}  \\
     &\,\,\, \text{s.t.} \,\,\,  \sum_{j=1}^{J^k} q(j)p_k(j) = 1, \, q(j) \geq 0, \\
     &\qquad  \sum_{j=1}^{J^k} \frac{1}{2} |q(j)p_k(j) - p_k(j)| \leq \alpha.
\end{split} \\
\begin{split}\label{eq:opt2}
    &= \min_{\boldsymbol{\lambda_1} \in \mathbb{R}^{J^k}, \lambda_2, \nu \in \mathbb{R}}\max_{\boldsymbol{q}}  \quad \mathcal{L}(\boldsymbol{\lambda}_1, \lambda_2, \nu) \\
    & \qquad \, \text{s.t.} \qquad \, \boldsymbol{\lambda}_1 \succeq 0, \, \lambda_2 \geq 0.
\end{split}
\end{align}
\end{subequations}
 where, $\mathcal{L}(\boldsymbol{\lambda}_1, \lambda_2, \nu)$ is the Lagrangian of the constrained optimization given in~\eqref{eq:opt1} given by,
\begin{align*}
    \mathcal{L}(\boldsymbol{\lambda}_1, \lambda_2, \nu) =& \sum_{j = 1}^{J^k} q(j) \bigg[ p_k(j) \mathbbm{1}_{\boldsymbol{x}^j_k \not\in \mathcal{X}} + \nu p_k(j) + \lambda_1(j) \bigg]  \\& -\lambda_2 \big(\sum_{j=1}^{J^k}p_k(j) \underbracket{|q(j) - 1|}_{f(q(j))} - 2\alpha\big) - \nu.
\end{align*}

The inner maximization of~\eqref{eq:opt2} can be solved by using convex conjugate of the function $f(x) = |x - 1|$ given by, 
$$
f^*(y) = \left\{ \begin{matrix} y && |y| \leq 1 \\ +\infty && |y|>1\end{matrix} \right. .
$$ Hence, we obtain $
    \max_{\boldsymbol{q}}\, \mathcal{L}(\boldsymbol{\lambda}_1, \lambda_2, \nu) = \lambda_2\sum_{j=1}^{J^k} p_k(j)f^*\big((\lambda_2p_k(j))^{-1} \small{\big(p_k(j) \mathbbm{1}_{\boldsymbol{x}^j_k \not\in \mathcal{X}} + \nu p_k(j) + \lambda_1(j)\big)\big)}  + 2\lambda_2\alpha - \nu$.

Now we substitute the above convex conjugate $f^*$ into~\eqref{eq:opt2}\footnote{We can take the inverse of $\lambda_2$ in the conjugate because the solution of~\eqref{eq:opt1} always lies on the boundary of the TVD constraint $\sum_{j=1}^{J} \frac{1}{2} |q(j)p_k(j) - p_k(j)| \leq \alpha$, i.e., the optimal $\lambda_2 > 0$.},
\begin{subequations}
\begin{align}
\begin{split}\label{eq:opt3}
     &\text{TVD}_{1-\alpha}(\mathbbm{1}_{\boldsymbol{x}_k \not\in \mathcal{X}})\\  &= \min_{\boldsymbol{\lambda_1}, \lambda_2, \nu} \, \sum_{j=1}^{J^k}\big(p_k(j) \mathbbm{1}_{\boldsymbol{x}^j_k \not\in \mathcal{X}} + \nu p_k(j) + \lambda_1(j)\big) + 2\lambda_2\alpha-\nu \\
     & \qquad \, \text{s.t.} \quad \,\boldsymbol{\lambda}_1 \succeq 0, \, \lambda_2 \geq 0, \\
     & \qquad -1 \leq (\lambda_2p_k(j))^{-1}\big(p_k(j) \mathbbm{1}_{\boldsymbol{x}^j_k \not\in \mathcal{X}} + \nu p_k(j) + \lambda_1(j)\big) \leq 1
\end{split}\\
\begin{split}\label{eq:opt4}
    &= \min_{\boldsymbol{\lambda_1}, \lambda_2, \nu} \, \mathbb{P}(\boldsymbol{x}_k \not\in \mathcal{X}) + \nu+  \sum_{j=1}^{J^k}\lambda_1(j) + 2\lambda_2\alpha -\nu\\
     & \qquad \, \text{s.t.} \quad \,\boldsymbol{\lambda}_1 \succeq 0, \, \lambda_2 \geq 0, \\
     & \qquad \qquad -\lambda_2p_k(j) - \nu p_k(j) - \lambda_1(j) \leq p_k(j) \mathbbm{1}_{\boldsymbol{x}^j_k \not\in \mathcal{X}}  , \\
     & \qquad \qquad \,\,\, p_k(j) \mathbbm{1}_{\boldsymbol{x}^j_k \not\in \mathcal{X}} \leq \lambda_2p_k(j)- \nu p_k(j) - \lambda_1(j) ,
\end{split} \\
\begin{split}\label{eq:opt5}
    &\leq \min_{\boldsymbol{\lambda_1}, \lambda_2, \nu} \,  \mathbb{P}(\boldsymbol{x}_k \not\in \mathcal{X})+  \sum_{j=1}^{J^k}\lambda_1(j) + 2\lambda_2\alpha\\
     & \qquad \, \text{s.t.} \quad \,\boldsymbol{\lambda}_1 \succeq 0, \, \lambda_2 \geq 0, \\
     & \qquad \qquad -\lambda_2p_k(j) - \nu p_k(j) - \lambda_1(j)\leq 0, \\
     & \qquad \qquad \,\,\,  p_k(j) \leq  \lambda_2p_k(j)- \nu p_k(j) - \lambda_1(j) ,
\end{split} \\
\begin{split}
&=  \mathbb{P}(\boldsymbol{x}_k \not\in \mathcal{X})+ {\alpha.}
\end{split}
\end{align}
\end{subequations}
In the above equations, we first substituted $f^*$ and constrained the argument of the conjugate to lie in [$-1, 1$] in~\eqref{eq:opt3} considering that the conjugate is unbounded outside this range. Afterwards, we noted that $$\mathbb{P}(\boldsymbol{x}_k \not\in \mathcal{X}) = \sum_{j=1}^{J^k}p_k(j) \mathbbm{1}_{\boldsymbol{x}^j_k \not\in \mathcal{X}},$$ and re-arranged the inequality constraints on $(\lambda_2p_k(j))^{-1}\big(p_k(j) \mathbbm{1}_{\boldsymbol{x}^j_k \not\in \mathcal{X}} + \nu p_k(j) + \lambda_1(j)\big)$ to obtain~\eqref{eq:opt4}. 

Next, we make the following constraints stricter in~\eqref{eq:opt5}
\begin{align*}
    -\lambda_2p_k(j) - \nu p_k(j) - \lambda_1(j)\leq 0\implies& \\ -\lambda_2p_k(j) - \nu p_k(j) -& \lambda_1(j) \leq p_k(j) \mathbbm{1}_{\boldsymbol{x}^j_k \not\in \mathcal{X}},\\
    p_k(j) \leq  \lambda_2p_k(j)- \nu p_k(j) - \lambda_1(j) \implies&\\  p_k(j) \mathbbm{1}_{\boldsymbol{x}^j_k \not\in \mathcal{X}} 
    \leq
    \lambda_2p_k(j)-& \nu p_k(j) - \lambda_1(j),
\end{align*}
to get a upper bound on~\eqref{eq:opt4} that is independent of the state $ {x}^j_k$. {Finally, we note that $(\boldsymbol{\lambda}_1, \lambda_2, \nu) = (\boldsymbol{0}, 0.5, -0.5)$ satisfies the KKT condition~\cite{boyd2004convex} and hence is the optimal solution to ~\eqref{eq:opt5} to complete the proof.}
\end{proof}
\begin{rmk}
{A natural interpretation of the tightening provided in Lemma~\ref{lem:risk_constraint} can be seen in Fig.~\ref{fig:illustration_dr}: one could have predicted that the probability of constraint violation can vary at most by $2\alpha$. The above proof provides a state-independent way to realize a similar tightening. It may be more conservative than the TVD constraint value, but it simply approximates the distributionally robust chance constraint by a more conservative chance constraint (with lower violation probability).} If $\alpha > \epsilon$, then the chance constraint becomes $\mathbb{P}(\boldsymbol{x}_k \not\in \mathcal{X}) < 0$, which is impossible to satisfy and indicates that the conservativeness of the chance constraint should be reduced by decreasing $\alpha$ or increasing $\epsilon$. 
\end{rmk}
\begin{lem}\label{lem:risk_tighten}
If Assumption~\ref{assumption: process_noise} holds, the TVD constraint~\eqref{eq:stcon} is satisfied if the following constraint is satisfied,
\begin{equation}\label{eq:cvar_stcon} 
 F_x\boldsymbol{\Tilde{x}}_k + \text{CVaR}_{\epsilon-\alpha}\big(F_x\boldsymbol{D}_{k}\boldsymbol{\bar{\delta}}_{k}\big) \leq g_x,
\end{equation}
where $\boldsymbol{\Tilde{x}}_k$ is the undisturbed nominal state: $\boldsymbol{x}_k = \boldsymbol{\Tilde{x}}_k + D\boldsymbol{\bar{\delta}}_k$.
\end{lem}
\begin{proof}
{In~\cite{nemirovski2007convex}, Nemirovski and Shapiro showed that CVaR provides a convex conservative approximation of the chance constraint. We have shown in Lemma~\ref{lem:risk_constraint} that~\eqref{eq:stcon} is satisfied if~\eqref{eq:strict_stcon} holds. Hence, a conservative approximation of $\mathbb{P}(\boldsymbol{x}_k \not\in \mathcal{X}) = \mathbb{P}(F_x\boldsymbol{x}_k - g_x > 0)\leq \epsilon - \alpha$ is given by,
\begin{align*}
    0 &\geq \text{CVaR}_{\epsilon-\alpha}(F_x\boldsymbol{x}_k - g_x) \\
        &= \text{CVaR}_{\epsilon-\alpha}(F_x\boldsymbol{\Tilde{x}}_k + F_x\boldsymbol{D}_k\boldsymbol{\bar{\delta}}_k- g_x)) \\
        &= F_x\boldsymbol{\Tilde{x}}_k - g_x + \text{CVaR}_{\epsilon-\alpha}( F_x\boldsymbol{D}_k\boldsymbol{\bar{\delta}}_k),
\end{align*}
where we obtain the first step by plugging in $\boldsymbol{x}_k = \boldsymbol{\Tilde{x}}_k + D\boldsymbol{\bar{\delta}}_k$ and the next step follows from the translational invariance property of coherent risk measures.}
\end{proof}
Lemmas \ref{lem:risk_constraint},~\ref{lem:risk_tighten} provide a simple tightening of the state constraints~\eqref{eq:stcon}. 

\begin{rmk}
CVaR can be computed using the minimization given in Section \ref{sec:preliminaries}. The chance constraint can be further tightened using the positive homogeneity property of coherent risk measures and i.i.d assumption on all disturbances,
\begin{equation}\label{eq:tighten_stcon}
    F_x\boldsymbol{\Tilde{x}}_k + \lVert F_xD\rVert_1\text{CVaR}_{\epsilon-\alpha}(|\delta|) \leq g_x.
\end{equation}
 {This tightening further reduces the size of the optimization problem, as $ \lVert F_xD\rVert_1\text{CVaR}_{\epsilon-\alpha}(|\delta|)$ can be expressed with approximately $ J^k$ fewer optimization variables and $2J^k$ fewer constraints for each $k\in\{0,\dotsc N-1\}$when computed online.}
\end{rmk}
\begin{lem}[Cost function]\label{lem:cost}
If the cost function given in~\eqref{eq:cost} is expressed as a quadratic cost with 
\begin{equation*}
    c(\boldsymbol{x}_k, \boldsymbol{u}_k) = \boldsymbol{x}_k^TQ\boldsymbol{x}_k + \boldsymbol{u}_k^TR\boldsymbol{u}_k,
\end{equation*}
then the MPC cost $\mathcal{J}(\boldsymbol{x}_t, \boldsymbol{u})$ is equivalently expressed as:
\begin{subequations}\label{eq:cost_new_val}
\begin{align}
\begin{split}
    &\small{\mathcal{J}(\boldsymbol{x}_t, \boldsymbol{u})} \\&= \small{\min_{m, z, \boldsymbol{s}}\!\!\sum_{k=0}^{N-1}\!\!\!c(\boldsymbol{\Tilde{x}}_k, \boldsymbol{u}_k)\! +\! \alpha m\! +\! (1-\alpha) \sum_{j=1}^{J^N}p_N(j) \big(z\!+\!\frac{s(j)}{1-\alpha} \big)}
    \end{split}\\
    \begin{split}\label{eq:cost_con1}
    & \small{\quad \textrm{s.t.} \quad m\geq\!\!\! \sum_{k=0}^{N-1}\!\!\!\big(\boldsymbol{D}_{k}\boldsymbol{\bar{\delta}}_{k}^j + 2A^{k}\boldsymbol{x}_0 +  2\boldsymbol{B}_{k}\boldsymbol{\bar{u}}_{k}\big)^TQ\boldsymbol{D}_{k}\boldsymbol{\bar{\delta}}_{k}^j,}
    \end{split}\\
    \begin{split}\label{eq:cost_con2}
    &\small{\qquad \quad\, s(j) + z \geq\!\!\!\sum_{k=0}^{N-1}\big(\boldsymbol{D}_{k}\boldsymbol{\bar{\delta}}_{k}^j + 2A^{k}\boldsymbol{x}_0 +  2\boldsymbol{B}_{k}\boldsymbol{\bar{u}}_{k}\big)^TQ\boldsymbol{D}_{k}\boldsymbol{\bar{\delta}}_{k}^j,}
    \end{split}\\
    \begin{split}\label{eq:cost_con3}
    &\qquad \quad s(j) \geq 0, \quad \forall j \in \{1, \dotsc, J^N \} 
    \end{split}
\end{align}
\end{subequations}
\end{lem}
\begin{proof}
See Appendix~\ref{appendix: lem3proof}.
\end{proof}
Using the reformulations afforded by Lemmas~\ref{lem:risk_constraint},~\ref{lem:risk_tighten},~\ref{lem:cost}, we can reformulate the MPC optimization given in~\eqref{eq:mpc1}.
\begin{thm}
If there exists a solution to the following quadratic program,
\begin{subequations}\label{eq:mpc2}
\begin{align}
\begin{split}\label{eq:cost_new}
\min_{\boldsymbol{u},m, z, \boldsymbol{s}}  &\small{\sum_{k=0}^{N-1}c(\boldsymbol{\Tilde{x}}_k, \boldsymbol{u}_k) + \alpha m + (1-\alpha) \sum_{j=1}^{J^N}p_N(j) \bigg(z +\frac{s(j)}{1-\alpha} \bigg)}
\end{split}\\
\begin{split}\label{eq:dyn_new}
 \textrm{s.t.} \quad &\boldsymbol{\Tilde{x}}_{k+1} = A\boldsymbol{\Tilde{x}}_k + B\boldsymbol{u}_k,
\end{split}\\
\begin{split}\label{eq:stcon_new}
&\small{F_x\boldsymbol{\Tilde{x}}_k + \text{CVaR}_{\epsilon-\alpha}\big(F_x\boldsymbol{D}_{k}\boldsymbol{\bar{\delta}}_{k}\big) \leq g_x, }
\end{split}\\
\begin{split}\label{eq:ic_new}
&\boldsymbol{u}_k \in \mathcal{U}, \,\quad\boldsymbol{\Tilde{x}}_0 = \boldsymbol{x}(t),\,
\end{split}\\
\begin{split}
   \eqref{eq:cost_con1},~\eqref{eq:cost_con2},~\eqref{eq:cost_con3}, \, \forall k \in \{0, \dotsc N-1 \}.
\end{split}
\end{align}
\end{subequations}
then the solution is a feasible solution of~\eqref{eq:mpc1}.
\end{thm}
\begin{proof}
We showed in Lemmas~\ref{lem:risk_constraint},~\ref{lem:risk_tighten} that satisfying~\eqref{eq:stcon_new} also satisfies~\eqref{eq:stcon}. We further showed in Lemma~\ref{lem:cost} that the cost function can be reformulated as a minimization. Plugging this cost function into the original MPC gives us a min-min optimization problem that can be combined into a one-layer optimization given by~\eqref{eq:mpc2}. This is true because the feasible solution to the one-layer optimization~\eqref{eq:mpc2} must be a feasible solution for the min-min problem and vice versa. Hence the optimal value of both optimizations must be equal.
\end{proof}
\section{Numerical Experiments}\label{sec:example}
{We compare our method, \textbf{DRMPC} given in~\eqref{eq:mpc2} and \textbf{tight DRMPC (TDRMPC)} that uses constraint~\eqref{eq:tighten_stcon} in place of~\eqref{eq:stcon_new}, against chance constrained stochastic MPC (SMPC) methods that evaluate the chance constraint using mixed integer variables as seen in~\cite{blackmore2010particleSMPC, luedtke2010integer} and a CVaR MPC (CMPC) approach inspired by~\cite{hakobyan2019risk,singh2018framework,Sopasakis2019}. In~\cite{hakobyan2019risk}, the authors consider a CVaR-constrained MPC whereas in~\cite{singh2018framework} the authors considered a CVaR cost. Similar to~\cite{Sopasakis2019}, although we don't consider dynamic risk, we consider CVaR cost and constraints in the MPC problem for the most consistent comparison to our method. The MPC optimizations considered for both these approaches are given below.
\vskip -0.2 true in
\begin{subequations}\label{eq:comp}
\begin{center}
\begin{align*}
\tag{\textbf{SMPC}}
\begin{split}
\min_{\boldsymbol{u}} \quad &\mathbb{E}\bigg(\sum_{k=0}^{N-1}c(\boldsymbol{x}_k, {u}_k)\bigg) \quad
\end{split}\\
\begin{split}
 \textrm{s.t.} \quad &\boldsymbol{x}_{k+1} = A\boldsymbol{x}_k + B\boldsymbol{u}_k + D\boldsymbol{{\delta}}_k,
\end{split}\\
\begin{split}
&\mathbb{P}(F_x\boldsymbol{x}_k - g_x > 0)\leq \epsilon, 
\end{split}\\
\begin{split}
&\boldsymbol{x}_0 = \boldsymbol{x}(t).
\end{split}\\\tag{\textbf{CMPC}}
\begin{split}
\min_{\boldsymbol{u}} \quad &\text{CVaR}_{1-\alpha}\bigg(\sum_{k=0}^{N-1}c(\boldsymbol{x}_k, \boldsymbol{u}_k)\bigg) \quad
\end{split}\\
\begin{split}
 \textrm{s.t.} \quad &\boldsymbol{x}_{k+1} = A\boldsymbol{x}_k + B\boldsymbol{u}_k + D\boldsymbol{{\delta}}_k,
\end{split}\\
\begin{split}
&\text{CVaR}_{\epsilon}(F_x\boldsymbol{x}_k - g_x)\leq 0, 
\end{split}\\
\begin{split}
&\boldsymbol{x}_0 = \boldsymbol{x}(t).
\end{split}
\end{align*}
\end{center}
\end{subequations}}

To illustrate the effectiveness and the advantages of the proposed method, we compare it to chance constrained stochastic MPC and CVaR MPC. We perturb the probability mass function of the disturbance to demonstrate the proposed method's distributional robustness.
We look at a simple two-dimensional discrete system $ \boldsymbol{x}_{k+1} = A\boldsymbol{x}_k + Bu_k + D\delta_k$, with
\begin{equation*}{\small
   A = \begin{bmatrix} 1.0475 & -0.0463 \\ 0.0463 & 0.9690\end{bmatrix}, \, B = D = \begin{bmatrix} 0.028 \\ -0.0195\end{bmatrix}.}
\end{equation*}
The control constraints are $-20 \leq u_k \leq 20,$ the state constraints are $-\begin{bmatrix} 4 & 4\end{bmatrix}^T \leq \boldsymbol{x}_k \leq \begin{bmatrix} 4 & 4\end{bmatrix}^T,$ and the disturbance lies in the set $\delta_k \in \{-1, \, 0,\, 1\}$ with probabilities $\boldsymbol{p} = \begin{bmatrix}0.1 & 0.8 & 0.1\end{bmatrix}$ respectively. We run 100 random simulations for each value of $\epsilon \in \{0.09,\,0.2,\,0.5,\,0.9\}$ such that each simulation has $35$ runs of the MPC optimization. The initial system state, $\boldsymbol{x}_0$, lies somewhere between $(3.1, 3.0)^T$ and $(4.1, 4.0)^T$. For each Monte-Carlo simulation, we randomly choose an initial condition in this range. 

\textbf{Discussion: }{The results are summarized in Table~\ref{table:1}. We have two comparisons for each value of $\epsilon$: the nominal case wherein we do not allow any perturbations to the original distribution $\boldsymbol{p}$ of the disturbances ($\alpha = 0$), and another allowing random variations in the distribution $\boldsymbol{p}$ with a total variation distance $\alpha$. As seen in Table~\ref{table:1}, for the 100 simulations, the percentage of violation of the constraints is consistently lower for the distributionally robust formulation. When $\alpha=0$, the chance constraint and cost for CMPC and DRMPC are equivalent and the results are the same for both. However, as soon as we allow for $\alpha$ perturbations in the distribution of the process noise, we see that DRMPC allows much fewer constraint violations and has a consistently lower cost. The TDRMPC is even more risk-averse than DRMPC due to further constraint tightening and we see the smallest percentage of constraint violations with behaviors that have consistently lower cost than SMPC and CMPC. }

 Fig.~\ref{fig:test} depicts one such Monte-Carlo simulation. The constraint $y_k \leq 4$, where $y$ is the second component of our state $\boldsymbol{x}_k$, is violated by both the SMPC and CVaR MPC controllers while the DRMPC and tight DRMPC controllers do not violate constraints. However, this risk-averse behavior comes at the cost of slower convergence to the origin.

The average times (in seconds) for each MPC iteration are
{\begin{align*}
\text{SMPC: } 0.32, \text{CMPC: } 0.47, \text{DRMPC: } 0.54 , \text{TDRMPC: } 0.11
\end{align*}}
\vskip -0.1 true in
run using YALMIP~\cite{yalmip} and a Gurobi solver~\cite{gurobi} in MATLAB (on a 2.7 GHz Quad-Core Intel Core i7 processor). { Thus, TDRMPC provides extra safety and reduced computational effort compared to SMPC and CMPC (with essentially no cost penalty in this example).  These results motivate and justify our risk-based chance constraint formulation with the novel constraint tightening approximations}. 


\begin{table}[t]
\vskip 0.1 true in
  \centering
  \resizebox{\columnwidth}{!}{%
  \begin{tabu}{|c|c|[1pt]c|c|[1pt]c|c|[1pt]c|c|[1pt]c|c|}
    \hline
    & $\boldsymbol{\epsilon}$ & \multicolumn{2}{|c|[1pt]}{\textbf{0.09}}&\multicolumn{2}{|c|[1pt]}{\textbf{0.2}} & \multicolumn{2}{|c|[1pt]}{\textbf{0.5}} & \multicolumn{2}{|c|}{\textbf{0.9}} \\[0.7ex]
    \hline
    & {$\mathit{\alpha}$}  & \textit{0} & \textit{0.05} & \textit{0} & \textit{0.15} & \textit{0} & \textit{0.4 } & \textit{0} & \textit{0.8}\\ [0.7ex]
    \tabucline[1pt]{1-10}
  \parbox[t]{2mm}{\multirow{4}{*}{\rotatebox[origin=c]{90}{Violations}}} &SMPC & 3.08& 3.17& 11.9& 12.2 & 14.6 & 15.9 & 23.4 & 24.1\\ [0.7ex]
   \tabucline[0.25pt]{2-10}
  & CMPC & 0 & 0 & 2.71 & 3.03  & 2.97 & 4.69 & 2.97  & 9.28\\ [0.7ex]
  \tabucline[0.25pt]{2-10}
  & DRMPC &0 &0 & 2.71 & 0 &2.97  & 0 & 2.97 & 0\\ [0.7ex]
  \tabucline[0.25pt]{2-10}
  & TDRMPC & 0 & 0  &  0.69 & 0  & 2.7 & 0 &  2.89& 0\\ [0.7ex]
   \tabucline[1pt]{1-10}
   \parbox[t]{2mm}{\multirow{4}{*}{\rotatebox[origin=c]{90}{{Cost}}}} & SMPC & 1.12& 1.13& 1.19& 1.19& 1.22& 1.21 &1.31 & 1.31\\[0.7ex]
   \tabucline[0.25pt]{2-10}
   & CMPC & 1.1 & 1.11  &  1.13& 1.12&1.19  & 1.18  & 1.21& 1.23 \\[0.7ex]
   \tabucline[0.25pt]{2-10}
   & DRMPC & 1.1& 1.1 & 1.13 & 1.01& 1.19 & 1.1 & 1.21 & 1.1\\ [0.7ex]
   \tabucline[0.25pt]{2-10}
    & TDRMPC & 1.01 & 1.02  &  1.01 & 1.02  & 1.12 & 1.02 &  1.18& 1.01 \\[0.7ex]
   \hline 
  \end{tabu}%
  }
  \caption{{Summary of results from Monte-Carlo simulations. The percentage of constraint violations and the average cost of each simulation $(\times 10^4)$ are compared.}}
  \vskip -0.2 true in
  \label{table:1}
\end{table}

\begin{figure}[tbp]
\centering
\vskip -0.15 true in
    \includegraphics[width=80mm,scale=0.5]{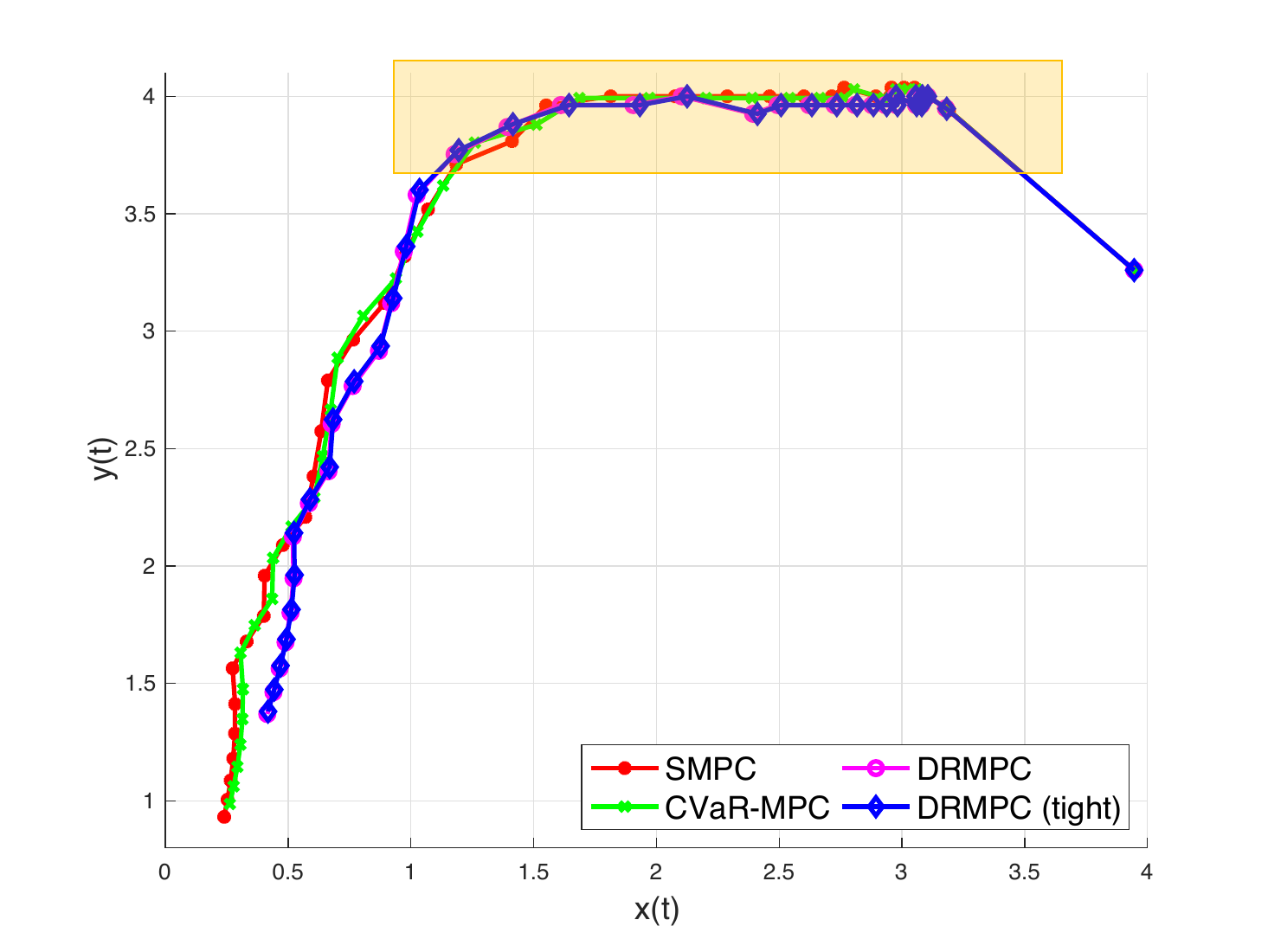}\\
     \vskip -0.15 true in
 \includegraphics[width=80mm,scale=0.5]{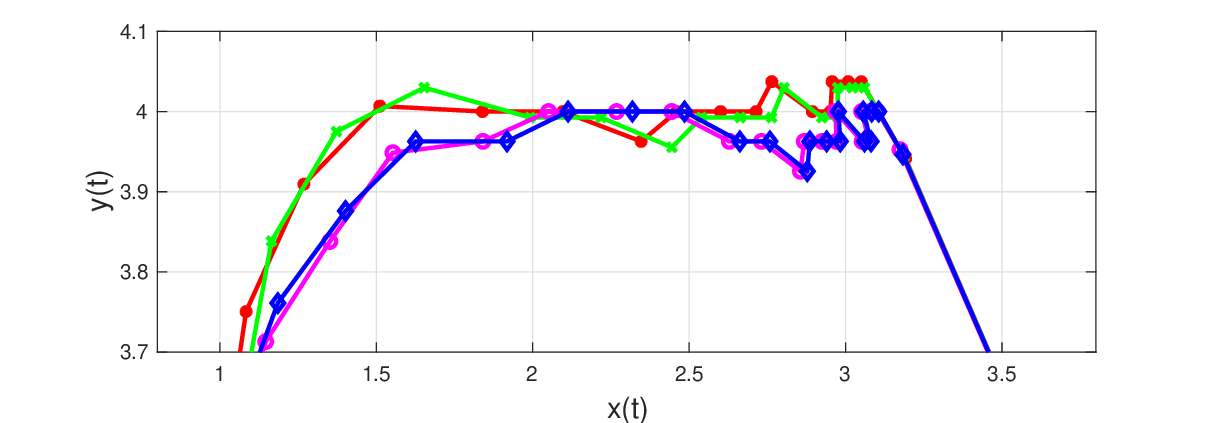}
\caption{\textbf{Top:} Comparison of the four controllers through visualization of one of the 100 simulations ($\epsilon = 0.5, \alpha = 0.4$). \textbf{Bottom:} Yellow region of the top figure zoomed in for clarity on the behavior near the boundary of the state constraint set, $y_k \le 4$.}
\vskip -0.2 true in
\label{fig:test}
\end{figure}
 
\section{Conclusions} \label{sec:conclusions}
This paper considered the problem of DRCC MPC with TVD-based ambiguity sets. We provided a CVaR tightening of the DRCC that can be computed offline, thereby reducing real-time computation. We used this DRCC in conjunction with a TVD-risk cost to get a distributionally robust MPC formulated as a quadratic program. Our simulation results showed the efficacy of our method in terms of both computational efficiency and safety in comparison to other methods. Future work involves extending this work to general coherent risk measures with f-divergence ambiguity sets. 
\vskip -0.5 true in

\appendix
\vskip -0.2 true in
\subsection{Batch matrices}\label{appendix: batch}
The state update matrices for batch form are given by,
\begin{align*}
    \boldsymbol{B}_k &= \begin{bmatrix}A^{k-1}B & A^{k-2}B &\dotsc & B\end{bmatrix} \in \mathbb{R}^{n_x\times kn_u}\\
    \boldsymbol{\bar{u}}_k &= \begin{bmatrix}\boldsymbol{u}_1^T &\boldsymbol{u}_2^T & \dotsc & \boldsymbol{u}_{k}^T\end{bmatrix}^T\in \mathbb{R}^{kn_u\times 1}\\
    \boldsymbol{D}_k &= \begin{bmatrix}A^{k-1}D & A^{k-2}D &\dotsc & D\end{bmatrix} \in \mathbb{R}^{n_x\times kn_d}\\
    \boldsymbol{\bar{\delta}}_k &= \begin{bmatrix}\boldsymbol{\delta}_1^T &\boldsymbol{\delta}_2^T & \dotsc & \boldsymbol{\delta}_{k}^T\end{bmatrix}^T\in \mathbb{R}^{kn_d\times 1}.
\end{align*}
\vskip -0.2 true in

\subsection{Proof of Lemma 3}\label{appendix: lem3proof}
Consider~\eqref{eq:cost} with a quadratic stage cost $c(\boldsymbol{x}_k, \boldsymbol{u}_k)$. 
\begin{align*}
    &\mathcal{J}(\boldsymbol{x}_0, \boldsymbol{u}) = \text{TVD}_{1-\alpha}\bigg(\!\sum_{k=0}^{N-1}c(\boldsymbol{x}_k, \boldsymbol{u}_k)\bigg) \\
    &= \small{\alpha\!\sup_{(\boldsymbol{\delta_1},\, \dotsc, \, \boldsymbol{\delta_k})}\sum_{k=0}^{N-1}\!\!c(\boldsymbol{x}_k, \boldsymbol{u}_k) + (1-\alpha)\text{CVaR}_{1-\alpha}\bigg(\!\sum_{k=0}^{N-1}c(\boldsymbol{x}_k, \boldsymbol{u}_k)\bigg)}
\end{align*}
Our goal is to find the worst-case value and the CVaR of the total stage cost.
Consider the nominal state update equation,
\begin{align}
    \boldsymbol{\Tilde{x}}_{k+1} = A^{k+1}\boldsymbol{x}_0 + \boldsymbol{B}_{k+1}\boldsymbol{\bar{u}}_{k+1}
\end{align}
We can write the quadratic stage cost as a function of the nominal state (without any disturbance effects) as,
\begin{align*}
    &c(\boldsymbol{x}_k, \boldsymbol{u}_k) = \boldsymbol{x}_k^TQ\boldsymbol{x}_k + \boldsymbol{u}_k^TR\boldsymbol{u}_k \\ 
    &= \small{(A^{k}\boldsymbol{x}_0 + \boldsymbol{B}_{k}\boldsymbol{\bar{u}}_{k} + \boldsymbol{D}_{k}\boldsymbol{\bar{\delta}}_{k})^TQ(A^{k}\boldsymbol{x}_0 + \boldsymbol{B}_{k}\boldsymbol{\bar{u}}_{k} + \boldsymbol{D}_{k}\boldsymbol{\bar{\delta}}_{k}) +\boldsymbol{u}_k^TR\boldsymbol{u}_k}\\
    &=(A^{k}\boldsymbol{x}_0 + \boldsymbol{B}_{k}\boldsymbol{\bar{u}}_{k})^TQ(A^{k}\boldsymbol{x}_0 + \boldsymbol{B}_{k}\boldsymbol{\bar{u}}_{k}) + 2(A^{k}\boldsymbol{x}_0)^TQ\boldsymbol{D}_{k}\boldsymbol{\bar{\delta}}_{k} \\& \qquad + 2(\boldsymbol{B}_{k}\boldsymbol{\bar{u}}_{k})^TQ\boldsymbol{D}_{k}\boldsymbol{\bar{\delta}}_{k} +(\boldsymbol{D}_{k}\boldsymbol{\bar{\delta}}_{k})^TQ\boldsymbol{D}_{k}\boldsymbol{\bar{\delta}}_{k} +\boldsymbol{u}_k^TR\boldsymbol{u}_k\\
    &=\boldsymbol{\Tilde{x}}_{k}^TQ\boldsymbol{\Tilde{x}}_{k} + \boldsymbol{u}_k^TR\boldsymbol{u}_k + (\boldsymbol{D}_{k}\boldsymbol{\bar{\delta}}_{k})^TQ\boldsymbol{D}_{k}\boldsymbol{\bar{\delta}}_{k} +\\ & \qquad 2(A^{k}\boldsymbol{x}_0)^TQ\boldsymbol{D}_{k}\boldsymbol{\bar{\delta}}_{k} +  2(\boldsymbol{B}_{k}\boldsymbol{\bar{u}}_{k})^TQ\boldsymbol{D}_{k}\boldsymbol{\bar{\delta}}_{k} \\
    &= c(\boldsymbol{\Tilde{x}}_k, \boldsymbol{u}_k) +\big(\boldsymbol{D}_{k}\boldsymbol{\bar{\delta}}_{k} + 2A^{k}\boldsymbol{x}_0 +  2\boldsymbol{B}_{k}\boldsymbol{\bar{u}}_{k}\big)^TQ\boldsymbol{D}_{k}\boldsymbol{\bar{\delta}}_{k}.
\end{align*}
We know from the translational invariance property of coherent risk measures that,
\begin{align*}
    \text{TVD}&_{1-\alpha}{\bigg(\!\sum_{k=0}^{N-1}c(\boldsymbol{x}_k, \boldsymbol{u}_k)\bigg)} =\!\! \sum_{k=0}^{N-1}c(\boldsymbol{\Tilde{x}}_k, \boldsymbol{u}_k) + \\ & \text{TVD}_{1-\alpha}\bigg(\!\sum_{k=0}^{N-1}\!\!\big(\boldsymbol{D}_{k}\boldsymbol{\bar{\delta}}_{k} + 2A^{k}\boldsymbol{x}_0 +  2\boldsymbol{B}_{k}\boldsymbol{\bar{u}}_{k}\big)^TQ\boldsymbol{D}_{k}\boldsymbol{\bar{\delta}}_{k}\bigg).
\end{align*}
Now, the above TVD is expressed as a combination of the worst-case value and the CVaR. The worst-case value is, 
\begin{align*}
    &\sup_{(\boldsymbol{\delta_1},\, \dotsc, \, \boldsymbol{\delta_k})}\sum_{k=0}^{N-1}\big(\boldsymbol{D}_{k}\boldsymbol{\bar{\delta}}_{k} + 2A^{k}\boldsymbol{x}_0 +  2\boldsymbol{B}_{k}\boldsymbol{\bar{u}}_{k}\big)^TQ\boldsymbol{D}_{k}\boldsymbol{\bar{\delta}}_{k} \\ 
    &= \min_{m\in\mathbb{R}} \quad m \\
    &\qquad \textrm{s.t.} \quad m \geq \sum_{k=0}^{N-1}\big(\boldsymbol{D}_{k}\boldsymbol{\bar{\delta}}_{k}^j + 2A^{k}\boldsymbol{x}_0 +  2\boldsymbol{B}_{k}\boldsymbol{\bar{u}}_{k}\big)^TQ\boldsymbol{D}_{k}\boldsymbol{\bar{\delta}}_{k}^j \\ & \qquad \qquad \qquad \qquad \forall j \in \{1, \dotsc, J^N \} .
\end{align*}
Similarly, the CVaR can be computed through a minimization as shown in Section~\ref{sec:preliminaries}. Hence the cost of the MPC can be rewritten as the minimization given by~\eqref{eq:cost_new_val}.
\vskip -0.2 true in
\footnotesize{
\bibliography{references}
}
\bibliographystyle{ieeetr}

\normalsize
\end{document}